\documentclass[sigconf, noacm]{acmart}
\pdfoutput=1
\AtBeginDocument{%
  \providecommand\BibTeX{{%
    \normalfont B\kern-0.5em{\scshape i\kern-0.25em b}\kern-0.8em\TeX}}}
\renewcommand\footnotetextcopyrightpermission[1]{}
\usepackage[ruled,vlined]{algorithm2e}
\usepackage[noend]{algpseudocode}
\usepackage{enumitem}
\usepackage{graphicx}
\usepackage{tikz}
\usepackage{subfigure}
\usepackage{pgfplots}
\pgfplotsset{compat=1.17}
\usepackage[font=small,skip=0pt]{caption}
\usepackage[most]{tcolorbox}
\graphicspath{ {./images/} }
\usepackage{balance}

\newtheorem*{conditions}{Conditions}

\begin{document}
\fancyhead{}

\title{Computing (1 + \texorpdfstring{$\epsilon$}{e})-Approximate Degeneracy in Sublinear Time}

\settopmatter{authorsperrow=3}
\author{Valerie King}
\affiliation{%
  \institution{University of Victoria}
  \streetaddress{P.O. Box 1212}
  \city{Victoria}
  \country{Canada}
  \postcode{43017-6221}
}
\email{val@uvic.ca}
\author{Alex Thomo}
\affiliation{%
  \institution{University of Victoria}
  \streetaddress{P.O. Box 1212}
  \city{Victoria}
  \country{Canada}
  \postcode{43017-6221}
}
\email{thomo@uvic.ca}
\author{Quinton Yong}
\affiliation{%
  \institution{University of Victoria}
  \streetaddress{P.O. Box 1212}
  \city{Victoria}
  \country{Canada}
  \postcode{43017-6221}
}
\email{quintonyong@uvic.ca}

\begin{abstract}
The problem of finding the degeneracy of a graph is a subproblem of the $k$-core decomposition problem. In this paper, we present a $(1 + \epsilon)$-approximate solution to the degeneracy problem which runs in $O(n \log n)$ time, sublinear in the input size for dense graphs, by sampling a small number of neighbors adjacent to high degree nodes. Our algorithm can also be extended to an $O(n \log n)$ time solution to the $k$-core decomposition problem.
This improves upon the method by Bhattacharya et al., which implies a $(4 + \epsilon)$-approximate $\tilde{O}(n)$ solution to the degeneracy problem, and our techniques are similar to other sketching methods which use sublinear space for $k$-core and degeneracy. We prove theoretical guarantees of our algorithm and provide optimizations, which improve the running time of our algorithm in practice. Experiments on massive real-world web graphs show that our algorithm performs significantly faster than previous methods for computing degeneracy, including the 2022 exact degeneracy algorithm by Li et al.

\end{abstract}





\maketitle

\section{Introduction}
The degeneracy of an undirected graph $G$, denoted $\delta$, is the smallest $\delta$ such that every induced subgraph of $G$ has a node with degree at most $\delta$. The problem of finding the degeneracy of a graph is a subproblem of the $k$-core decomposition problem, where a $k$-core is a maximal connected subgraph in which all nodes have degree at least $k$, and $k$-core decomposition problem is to find all $k$-cores of a given graph $G$. The degeneracy of $G$ is the maximum $k$ for which $G$ contains a (non-empty) $k$-core.

The $k$-core decomposition problem, and consequently the degeneracy problem, can be solved using a well known linear time ``peeling" algorithm \cite{peeling}. Starting with $k=1$, the algorithm repeatedly removes all nodes in $G$ with degree less than $k$ and their incident edges, until the only remaining nodes have degree at least $k$ and their connected components are $k$-cores. Then $k$ is incremented, and the algorithm is repeated on the remaining subgraph. There has been much work done to improve the practical running time of $k$-core decomposition fully in-memory \cite{kcore1}, with fully external-memory \cite{kcore2}, and with semi external-memory \cite{kcore3}. There are also algorithms for approximate computing of $k$-core in dynamic, streaming, and sketching settings \cite{dynamickcore1, streamingkcore1, kcoresketching}. 

While $k$-core algorithms can be used to compute degeneracy, they may incur unnecessary costs and space for processing cores for small $k$. There are algorithms which directly compute the degeneracy of the graph. The SemiDeg+ algorithm from \cite{semideg} is the state-of-the-art in exact degeneracy computation. In a semi-streaming model, \cite{colton1}, \cite{semistream1}, and \cite{semistream2} propose $(2 + \epsilon)$-approximate solutions and \cite{colton2} proposes a $(1+\epsilon)$-approximate solution. These require each edge to be scanned. 

The degeneracy of a graph is used for many applications in graph analysis and graph theory including measuring the sparcity of a graph, approximating dominating sets \cite{domset1, domset2}, approximating the arboricity of a graph within a constant factor \cite{arboricity1}, and maximal clique enumeration \cite{maxclique1}. The degeneracy is at least half the maximum density of any subgraph \cite{colton1}. A sublinear approximate algorithm for degeneracy is the algorithm presented in \cite{sublineardense} for computing maximum density of any subgraph. It runs in $\tilde{O}(n)$ time and implies a $(4+\epsilon)$-approximate algorithm for degeneracy.

In this work, we present a $(1 + \epsilon)$-approximate $O(n \log n)$ time solution to the degeneracy problem where $n$ is the number of nodes in the graph. Hence, our algorithm is sublinear (in the number of edges for dense graphs).  Our algorithm can be extended to a sublinear time algorithm for $k$-core decomposition, but the time savings over other algorithms are particularly exhibited empirically for the degeneracy problem.  We avoid the overhead of full $k$-core decomposition and directly compute degeneracy of the given graph $G$ by sampling the neighbors of nodes starting with those with the highest degree above a certain threshold $l$. We then determine if the subgraph induced by those nodes contains an $l$-core by repeatedly eliminating nodes which did not sample a sufficient number of other high degree nodes, similar to the peeling algorithm. If it does, then we output $l$ as the (approximate) degeneracy of $G$. We prove the theoretical guarantees of our algorithm and we also propose optimizations which improve its practical running time. We also compare our solution to the $(1 + \epsilon)$-approximate algorithm of \cite{colton2} and the SemiDeg+ algorithm from \cite{semideg} on massive real-world webgraphs and show that our algorithm is substantially faster than both on almost all datasets. We show the following theorem:
\begin{theorem} \label{t:main}
There is an algorithm for $k$-core decomposition which runs in time 
$O(\min\{m, n \log n\})$ and gives a $(1+\epsilon)$ factor approximation for the core values of a graph with $n$ nodes and $m$ edges in the incident-list model, with high probability. 
\end{theorem}

\section{Preliminaries}
\subsubsection*{Definitions and Notation}
Let $G=(V, E)$ be an undirected graph with nodes $V$ and edges $E$, where $|V|=n$ and $|E|=m$. We denote the degree of a vertex $v\in V$ in $G$ as $\deg(v)$.  A subgraph of $G$ induced by a subset of nodes $V' \subseteq V$ is defined as $G_{V'}=(V', E')$ where $E'=\{(u, v) \mid u, v \in V', (u, v) \in E\}$.  We denote the degree of a vertex $v$ in $G_{V'}$ as $d_{G_{V'}}(v)$.

For a graph $G$ and an integer $k$, a $k$-core of G is the maximal induced connected subgraph of $G$ such that every node in the subgraph has degree at least $k$ \cite{kcoredef}. The core number of a node $v \in V$ is the largest $k$ for each there is a $k$-core that contains $v$. The degeneracy $\delta$ of a graph $G$ is defined as the smallest integer such that every non-empty induced subgraph has a vertex with degree at most $\delta$ \cite{degeneracydef}. The degeneracy of $G$ is also equal to the largest $k$ for which $G$ contains a non-empty $k$-core \cite{maxclique1} (i.e. the degeneracy equals the maximum core number). 

\subsubsection*{Problem Definition} Our goal is to design an efficient algorithm for computing an approximation of the degeneracy $\delta$ of a given graph $G$ within a factor of $(1 + \epsilon)$ for $\epsilon \in (0, 1]$, where $\epsilon$ is an input parameter, with high probability. We define high probability as occurring with probability at least $1 - \frac{O(1)}{n^c}$, where $c$ is a constant given as an input parameter.

\subsubsection*{Graph Input Model} The graph input model we use is the incident-list model, where each node $v \in V$ has an incidence list of the neighbors of $v$, arbitrarily ordered. This is a standard graph representation used in sublinear time algorithms \cite{sublineardense, sublinearmodel1}. There are two types of queries to access graph information in this model:
\begin{enumerate}
    \item degree queries to get the degree of any node $v \in V$
    \item neighbor queries to get the $i^{th}$ neighbor of any node $v \in V$, which is the $i^{th}$ element in $v$'s incidence list
\end{enumerate}

\section{Related Work}
In 1983, Matula and Beck \cite{peeling} presented the linear time peeling algorithm for $k$-core decomposition. Although it is superseded in practice by heuristic methods in other work and presents unnecessary overhead for degeneracy computation (as do other algorithms which enumerate all $k$-cores), it is an efficient method for core computation especially for small graphs and graphs with small degeneracy. 

\cite{kcoresketching} and \cite{colton2} present $(1 + \epsilon)$-approximate solutions for $k$-core decomposition and degeneracy computation respectively. The goal of both algorithms is to minimize memory usage in a streaming / semi-streaming model, which make them inefficient in terms of running time in the incident-list model due to the requirement of making a random choice for every edge in the graph, especially with large graphs. However, their techniques are related to our work.

The work in \cite{sublineardense} proposes a $(2 + \epsilon)$-approximate solution to the densest subgraph problem in sublinear time, which implies a $(4 + \epsilon)$-approximate sublinear algorithm for degeneracy, and \cite{sublineardense} uses the same model as ours, the incident-list model, for their graph input representation. The SemiDeg+ algorithm from 2022 \cite{semideg} is the state-of-the-art in practical exact degeneracy computation, so we use it as a comparison baseline for the experiments on the running times of our approximate algorithm.

\section{Degeneracy Algorithm}
\subsubsection*{High Level Description}
The key insight behind the algorithm is that the maximum core of the graph will contain high degree nodes which are connected to many other high degree nodes.
So, we can efficiently compute the approximate degeneracy of the graph by looking at the subgraph $H$ induced by the high degree nodes.
Initially, the nodes in $V$ are put into $H$ if their degree is above a certain threshold $l$, where $l$ is what we approximate as the degeneracy of the graph, and we decrease $l$ in each iteration when there is a high probability that $H$ does not contains an $l$-core (within an error factor of $(1 + \epsilon)$).
To efficiently check if there is a high probability that $H$ does contain an $l$-core, we sample a small number of neighbors incident to each node in $H$. If a node $v \in H$ is in the $l$-core, then it is likely that it samples a sufficient fraction of its number of neighbors that are also in $H$.
If a node does not sample a sufficient number of neighbours in $H$, then it is likely not part of an $l$-core,
and we remove the node from $H$ and recursively peel its neighbors from $H$ similar to the peeling algorithm. If no nodes remain unpeeled, then we lower the threshold and repeat the procedure. Otherwise, we conclude that $H$ has an approximate core number of $l$, and we output $l$ as the approximate degeneracy of $G$. In Section \ref{sec:theoreticalguarantees}, we prove that the output is within a factor of $(1 + \epsilon)$ of the true degeneracy of $G$ with probability $\geq 1-\frac{2}{n^c}$.

\subsubsection*{Algorithm Details}
The approximate algorithm for computing degeneracy is illustrated in Algorithm \ref{alg:approxalgorithm}. In each iteration of the while-loop at Line 3, we partition the nodes in $V$ into $H$ and $L$, where $H$ contains the nodes which have a degree greater than or equal to a certain threshold $l$ and $L=V\setminus H$ (Line 4). Initially, $l=\frac{n}{(1+\epsilon_1)}$, where $\epsilon_1 = \epsilon/ 3$, and we decrease $l$ by a factor of $(1 + \epsilon_1)$ (Lines 1 and 27). 

Then, we sample a small number of neighbors for each node $v \in H$ (Line 10). We sample $k(v)$ neighbors $u$ of $v$ with replacement for every node $v \in H$, where $k(v)=\lceil p \deg(v) \rceil$. Initially $p =\frac{2((1+c)\ln(n) + \ln(\log_{1+\epsilon_1}n))(1 + \epsilon_1)^2}{\epsilon_1^2 n}$ and we increase $p$ by a factor of $(1+\epsilon_1)$ in each iteration (Lines 1 and 28). 

Next, we check if $H$ contains an $l$-core. This is done by a process similar to peeling, in which we recursively peel nodes from $H$ which did not sample a sufficient number of neighbors in $H$. In particular, each node $v \in H$ is assigned a value $t(v)$ which is initially set to the number of neighbors sampled by $v$, $k(v)$, and represents the degree of $v$ in $H$ (Line 8). Then, since we sampled $k(v)/\deg(v)$ neighbors, we expect that at least $l \cdot k(v) / \deg(v)$ edges incident to other nodes in $H$ were sampled if the current nodes in $H$ is an $l$-core. So, if at any point $t(v)$ for a node $v \in V$ is less than $\frac{l \cdot k(v)}{\deg(v)}$, then $v$ is peeled from $H$ to $L$ (Lines 13 and 22).

During the sampling process (Lines 9 - 17), whenever a node $v \in H$ samples a neighbour $u \in L$, we decrement $t(v)$ (Line 12). Otherwise, if it samples a neighbour $u \in H$, we need to remember that $u$ was sampled by $v$ so if $u$ is peeled from $H$ later, we also decrement $t(v)$. Each node in $u \in H$ maintains an initially empty list (with possible repeats) $Sampled(u)$ of nodes which sampled them. So, if $v \in H$ samples a neighbor $u \in H$, we add $v$ to $Sampled(u)$ (Line 17). Then, if a node in $u \in H$ is peeled due to $t(u)$ falling below $\frac{l \cdot k(u)}{deg(u)}$, we add $u$ to an initially empty queue $Q$ (Line 15). 

After the sampling process is finished, we remove each node $u$ from $Q$, iterate through each node $v \in Sampled(u)$, and decrement $t(v)$ if $v$ is still in $H$ (Line 22). Similarly, if $t(v)$ for a node $v \in H$ is less than $\frac{l \cdot k(v)}{\deg(v)}$, then $v$ is moved to $L$ and then is also added to $Q$ (Lines 24 and 25). If there are no remaining nodes in $H$, that means each node in $H$ did not sample a sufficient number of other nodes in $H$, and it is likely that there is no $l$-core in $G$. So, we continue to the next iteration where we repeat the process with $l = l/(1+\epsilon_1)$, and $p = p \cdot (1 + \epsilon_1)$, and we also set the sampled lists to empty (Lines 27 - 29). Otherwise, if there are nodes remaining in $H$, then we expect that $H$ is at least an $l$-core and $l$ is the highest threshold for which we found a core thus far. Therefore, we return $l$ as an approximation to the degeneracy of $G$ (Line 31). In the case where $p \geq 1$ before $l$ is found, we simply compute $\delta$ by running the peeling algorithm and return $\delta$.

\begin{algorithm}
\caption{ApproximateDegeneracy($G, \epsilon, c$)}
\small
\textbf{Input:} Graph $G=(V, E)$, error factor $\epsilon \in (0, 1]$, constant $c$\\
\textbf{Output:} Approximate degeneracy $\delta$ of $G$ within a factor of $(1 + \epsilon)$ with probability $1-\frac{2}{n^c}$
\begin{algorithmic}[1]
\State $\epsilon_1 = \epsilon / 3, l=\frac{n}{(1 + \epsilon_1)}, p = \frac{2((1+c)\ln(n) + \ln(\log_{1+\epsilon_1}n))(1 + \epsilon_1)^2}{\epsilon_1^2 n}$
\State $Sampled(v) \leftarrow \emptyset$ for all $v \in V$
\State \textbf{while} $p < 1$ \textbf{do}
\State \indent $H \leftarrow \{v \in V \mid \deg(v) \geq l\}, L \leftarrow V \setminus H$
\State \indent $Q \leftarrow \emptyset$
\State \indent \textbf{for} each node $v \in H$ \textbf{do}
\State \indent \indent $k(v)\leftarrow \lceil p \deg(v) \rceil$
\State \indent \indent $t(v) \leftarrow k(v)$
\State \indent \indent \textbf{for} $i \in {1, ... , k(v)}$ \textbf{do}
\State \indent \indent \indent pick a neighbor $u$ of $v$ independently at random \label{line:sample}
\State \indent \indent \indent \textbf{if} $u \in L$ \textbf{then}
\State \indent \indent \indent \indent $t(v) = t(v) - 1$
\State \indent \indent \indent \indent \textbf{if} $t(v) < \frac{l \cdot k(v)}{\deg(v)}$ \textbf{then}
\State \indent \indent \indent \indent \indent $H \leftarrow H \setminus \{v\}, L \leftarrow L \cup \{v\}$
\State \indent \indent \indent \indent \indent add $v$ to $Q$
\State \indent \indent \indent \textbf{else}
\State \indent \indent \indent \indent add $v$ to $Sampled(u)$
\State \indent \textbf{while} $Q \neq \emptyset$ \textbf{do}
\State \indent \indent $u \leftarrow dequeue(Q)$
\State \indent \indent \textbf{for} each node $v \in Sampled(u)$ \textbf{do}
\State \indent \indent \indent  \textbf{if} $v \in H$ \textbf{then}
\State \indent \indent \indent \indent $t(v) = t(v) - 1$
\State \indent \indent \indent \indent \textbf{if} $t(v) < \frac{l \cdot k(v)}{\deg(v)}$ \textbf{then}
\State \indent \indent \indent \indent \indent $H \leftarrow H \setminus \{v\}, L \leftarrow L \cup \{v\}$
\State \indent \indent \indent \indent \indent add $v$ to $Q$
\State \indent \textbf{if} $H = \emptyset$ \textbf{then}
\State \indent \indent $l \leftarrow l / (1 + \epsilon_1)$
\State \indent \indent $p \leftarrow  p \cdot (1 + \epsilon_1)$
\State \indent \indent $Sampled(v) \leftarrow \emptyset$ for all $v \in V$
\State \indent \textbf{else}
\State \indent \indent \textbf{return} $l$ \label{line:return}
\State compute $\delta$ by running the peeling algorithm
\State \textbf{return} $\delta$
\end{algorithmic}
\label{alg:approxalgorithm}
\end{algorithm}

\section{Theoretical Guarantees}
\label{sec:theoreticalguarantees}

In this section, we will prove the probability of correctness and sublinear running time of Algorithm \ref{alg:approxalgorithm}.

\subsection{Correctness}
\label{subsec:degeneracycorrectness}
Let $(H, L)$ be a partition of the vertices of $G$. Initially $H = V$ and $L=\emptyset$. Let $d_H(v)$ denote the degree of a vertex $v \in H$ in the subgraph induced by the nodes in $H$ such that the following conditions hold for some $d_1$ and $d_2$:

\begin{conditions}$ $
\begin{enumerate}
    \item each node in $v \in H$ is always moved to $L$ if $d_H(v) < d_2$ 
    \item a node $v \in H$ is not moved to $L$ if $d_H(v) \geq d_1$
\end{enumerate}
\end{conditions}

Lemma \ref{lem:uniformsample} enables us to regard each sample as an independent uniformly random coin flip, whose outcome is heads if the neighbor $u$ incident to the sampled edge is in $H$, and the probability of heads is $d_H(v)/\deg(v)$. 

\begin{lemma}
\label{lem:uniformsample}
At any time during an iteration, the event that, in a single execution of Line 10, the neighbor picked by $v$ is in $H$ has probability $d_H(v)/\deg(v)$, and the set of events corresponding to the $k(v)$ executions of Line 10 are mutually independent.
\end{lemma}
\begin{proof}
Consider $H$ at a specific time.
Let $E_i$ be the event that $v$'s $i^{th}$ sample (Line 10) is in $H$.
The event that any one neighbor is selected is an independent uniformly random event with probability with probability $1/\deg(v)$. These random choices are independent from the state of $H$ at any time and they are independent of each other since they are picked with replacement. Thus, at any point in the algorithm, if we fix the nodes in $H$, the probability of any node $v$ sampling a neighbor in $H$ is exactly the number of its neighbors in $H$, $d_H(v)$, over its degree. 
\end{proof}

Next, Lemma \ref{lem:t(v)value} states the value of $t(v)$ at any time for any node $v \in H$.

\begin{lemma}
\label{lem:t(v)value}
At any time, for any node $v \in H$, $t(v)$ is the number of $v$'s sampled neighbors in $H$. 
\end{lemma}
\begin{proof}
For each node $v \in H$, $t(v)$ is initialized to the total number of samples, $k(v)$. Whenever $v$ discovers that a sampled neighbor is not in $H$, either immediately in Line 11 or later in Line 21, $t(v)$ is decremented. So, $t(v) = k(v)\text{ } - $ the number of sampled neighbors in $L$.
\end{proof}

Then, we have that Lemma \ref{lem:expectedt(v)} follows directed from Lemma \ref{lem:t(v)value}.

\begin{lemma}
\label{lem:expectedt(v)}
The expected value of $t(v)$ for each node $v \in H$ is $E[t(v)]=\frac{d_H(v)}{\deg(v)}k(v)$.
\end{lemma}

Next, we relate $d_H(v)$ to the core values of nodes in $L$ and $H$.

\begin{lemma}
\label{lem:hconditions}
Assuming conditions (1) and (2) above hold, at Line 26, the nodes in $H$ have core values at least $d_2$ and the nodes in $L$ have core values less than $d_1$.
\end{lemma}
\begin{proof}
By condition (1) above, every node in $H$ has $\geq d_2$ neighbors in $H$. So, every node in $H$ must have a core value of at least $d_2$. 

Suppose for a contradiction that a node is moved to $L$ with a core value $\geq d_1$. Then there must be a first node $u$ for which this happens. Since the core value of $u$ is $\geq d_1$ by the assumption, $u$ must have $ \geq d_1$ neighbors with core values at least $d_1$. Since $u$ is the first such node to be moved to $L$ with core value $\geq d_1$, none of $u$'s neighbors which are in the $d_1$-core have been moved to $L$. Hence, $d_H(u) \geq d_1$ and $u$ cannot be moved to $L$ by condition (2).
\end{proof}

Now, we will prove that the bound $\frac{l \cdot k(v)}{\deg(v)}$ (Line 13 and 23), which we require for $t(v)$ for a node $v$ to remain in $H$, gives us an approximate bound for the degree of $v$ in $H$.

\begin{lemma}
\label{lem:probbounds}
When a node is moved to $L$, for every node $v \in H$
\begin{enumerate}
    \item if $d_H(v) \geq l \cdot (1 + \delta_1)$, then $t(v) \geq \frac{l \cdot k(v)}{\deg(v)}$, and
    \item if $d_H(v) < l / (1 + \delta_2)$, then $t(v) < \frac{l \cdot k(v)}{\deg(v)}$
\end{enumerate}
with high probability, where $\delta_1 = \epsilon_1$ and $\delta_2=\frac{3}{2}\epsilon_1$ for $\epsilon_1 \leq 1$.
\end{lemma}
\begin{proof}
By Lemma \ref{lem:uniformsample}, each sample can be regarded as an independent random coin flip which is heads when a sampled neighbor is in $H$. We will use the following Chernoff bounds from Theorem 2.1 in \cite{chernoffbound}. Let $X$ be the sum of the number of heads and let $\mu = E[X]$. For $0 < \delta_1, \delta_2$
\begin{equation}
\label{eq:chernoff1}
    \Pr(X < (1 - \delta_1)\mu') \leq e^{-\frac{-\delta_1^2\mu'}{2}} \text{ for all } \mu' \leq \mu 
\end{equation}
and
\begin{equation}
\label{eq:chernoff2}
    \Pr(X \geq (1 + \delta_2)\mu') \leq e^{-\frac{-\delta_2^2\mu'}{2+\delta_2}} \text{ for all } \mu' \leq \mu 
\end{equation}

In our experiment, we have $k(v)$ coin flips where the probability of heads is $d_H(v)/\deg(v)$ and by Lemma \ref{lem:expectedt(v)}, $E[t(v)]=\frac{d_H(v)}{\deg(v)} k(v) = \mu$.

To prove (1), when $d_H(v) \geq l \cdot (1 + \delta_1)$, let $\mu'=(1 + \delta_1) \cdot l \cdot k(v) / \deg(v)$. If $d_H(v) \geq l \cdot (1 + \delta_1)$, then $E[t(v)] \geq \mu'$. Now, we can bound the probability that $t(v) < l \cdot k(v) / \deg(v)$, which we can rewrite as $t(v) < (1 - \frac{\delta_1}{1 + \delta_1})\mu'$. By the Chernoff bound in Equation \ref{eq:chernoff1}, the probability that $t(v) < (1 - \frac{\delta_1}{1 + \delta_1})\mu'$ is $ \leq e^{\frac{-\delta_1^2\mu'}{2(1 + \delta_1)^2}}=e^{\frac{-\delta_1^2 \cdot  l \cdot k(v) / \deg(v)}{2(1 + \delta_1)}}$.

To prove (2), when $d_H(v) < l / (1 + \delta_2)$, let $\mu'=\frac{l}{(1+\delta_2)}\frac{k(v)}{\deg(v)} > \frac{d_H(v)}{\deg(v)}k(v) = \mu$. By the Chernoff bound in Equation \ref{eq:chernoff2}, $\Pr(t(v)) \geq l \cdot k(v) / \deg(v) = \Pr(t(v) \geq (1 + \delta_2)\delta') \leq e^\frac{-\mu'\delta_2^2}{2+\delta_2}=e^\frac{-\delta_2^2l \cdot k(v)/\deg(v)}{(1 + \delta_2)(2 + \delta_2)}$.

Since $k(v) \geq p \cdot \deg(v)$, we have $l \cdot k(v) / \deg(v) \geq l \cdot p$. So, the probability that (1) fails to hold is $\leq e^\frac{-\delta_1^2 l\cdot p}{2(1 + \delta_1)}$ and the probability that (2) fails to hold is $\leq e^\frac{\delta_2^2 l \cdot p}{(1 + \delta_2)(2 + \delta_2)}$. We let $\delta_2=(3/2)\delta_1$ to get that bound on the probability of error (2) is $e^\frac{-(9/4)\delta_1^2 l \cdot p}{2 + (9/2)\delta_1+(9/4)\delta_1^2} \leq e^\frac{-\delta_1^2 l \cdot p}{2 (1 + \delta_1)}$ when $\delta_1 \leq 1$. Recall that $\delta_1=\epsilon_1$. 

Then, we take the union bound of both error probabilities (1) and (2) over all $n$ nodes. Since $l$ is initialized to $n / (1 + \epsilon_1)$, and we decrease $l$ by a factor of $(1 + \epsilon_1)$ in each iteration, the algorithm will take less than $\log_{1 + \epsilon_1}n$ iterations. So, we multiply both errors (1) and (2) by $e^{\ln n + \ln \log_{1 + \epsilon_1}n}$ for the union bound and over all possible iterations. Therefore, the probability each error is less than $1/n^c$ when $l \cdot p \geq \frac{2((1 + c)\ln n + \ln \log_{1 + \epsilon_1}n)(1 + \epsilon_1)}{\epsilon_1^2}$ and the probability any error occurs is less than $2/n^c$ for any $c$.
\end{proof}

Next, we show that Algorithm \ref{alg:approxalgorithm} satisfies the assumptions of Lemma \ref{lem:hconditions} with high probability.

\begin{lemma}
\label{lem:preconditionsproof}
In any iteration of Algorithm \ref{alg:approxalgorithm},
\begin{enumerate}
    \item every node $v \in H$ is moved to $L$ when $d_H(v) < l / (1 + (3/2)\epsilon_1)$
    \item no node $v \in H$ is moved to $L$ when $d_H(v) \geq l \cdot (1 + \epsilon_1)$
\end{enumerate}
with high probability.
\end{lemma}
\begin{proof}
Let $d_1 = l\cdot(1 + \epsilon_1)$ and $d_2=l / (1 + (3/2)\epsilon_1)$. By Lemma \ref{lem:probbounds}, with high probability, no node $v \in V$ with $d_H(v) > d_1$ has $t(v) < l \cdot k(v) / \deg(v)$, so $v$ is not moved to $L$. Also, every node $v \in H$ with $d_H(v) < l / (1 + (3/2)\epsilon_1)$ has $t(v) < l \cdot k(v) / \deg(v)$, so $v$ is moved to $L$.
\end{proof}

We can now prove the bounds on the approximate degeneracy output by Algorithm \ref{alg:approxalgorithm}.

\begin{lemma}\label{lem:delta}
Let $l$ be the output returned by the algorithm in line \ref{line:return} in the case a value is returned. With high probability, for each $v \in V$, $c(v) < l \cdot (1 + \epsilon_1)^2$ and for all $v\in H$, $c(v) \geq l /(1 + (3/2)\epsilon_1)$. In particular, $\delta / (1 + \epsilon_1)^2 < l \leq \delta(1 + (3/2)\epsilon_1)$. 
\end{lemma}
\begin{proof}
In the iteration when the threshold was $l \cdot (1 + \epsilon)$, all nodes in $H$ were moved to $L$. By Lemma \ref{lem:preconditionsproof}, at the end of the algorithm, when the threshold is $l$, $d_H(v)$ for each node $v \in H$ satisfies $d_H(v) \geq  l / (1 + (3/2)\epsilon_1)$ and since $v$ was moved to $L$ in the previous iteration, $d_H(v) <l\cdot (1 + \epsilon_1)^2$. Let $c(v)$ denote the core number of a node $v$. Then by Lemma \ref{lem:hconditions}, for each $v \in H$ we have $c(v) < l \cdot (1 + \epsilon_1)^2$ and $c(v) \geq l / (1 + (3/2)\epsilon_1)$, which gives us the bounds $c(v) / (1 + \epsilon_1)^2 < l \leq c(v)(1 + (3/2)\epsilon_1)$. Since $l$ is the highest threshold for which not all nodes in $H$ were moved to $L$, it follows that $\delta / (1 + \epsilon_1)^2 < l \leq \delta(1 + (3/2)\epsilon_1)$.

\end{proof}

We observe that if we set $\epsilon_1 = \epsilon/3$, we get  $(1 + \epsilon/3)^2=1 + 2\epsilon/3 + \epsilon^2/9$ which is less than $1 + \epsilon$ when $\epsilon \leq 1$. Also, if no $l$ is returned and $p\geq 1$, we run the peeling algorithm and return $\delta$. Thus, Algorithm \ref{alg:approxalgorithm} gives a $(1 + \epsilon)$ approximation.

\begin{lemma}
\label{lem:correctness}
Given a graph $G$, Algorithm \ref{alg:approxalgorithm} outputs a $(1 + \epsilon)$ approximation to the degeneracy $\delta$ of $G$.
\end{lemma}

\subsection{Running Time}
\label{subsec:degeneracyrunningtime}
Let $j$ be the number of iterations (of the while-loop at Line 3) which Algorithm \ref{alg:approxalgorithm} takes to compute the approximate degeneracy of an input graph $G$. When the output $l=n/(1 + \epsilon_1)^j$, then $j = \log_{1+\epsilon_1}(n/l)$. Consider the cost of sampling the neighbors of any node $v \in H$. Let $k_i(v)$ and $p_i$ be the number of neighbors sampled and the value of $p$, respectively, in any iteration $i$ of the algorithm. The number of neighbors of $v$ which are sampled over the algorithm is: 
\begin{align*}
    &\leq \sum_{i=1}^j k_i(v) = \sum_{i=1}^j \lceil p_i \deg(v) \rceil \\
    &\leq \sum_{i=1}^j (p_i \deg(v) + 1) = \log_{1 + \epsilon_1}(n/l) + \sum_{i=1}^j p_i \deg(v)\\
    &= \log_{1 + \epsilon_1}(n/l) + \tfrac{2((1+c)\ln(n) + \ln(\log_{1+\epsilon_1}n))(1 + \epsilon_1)^2}{\epsilon_1^2 n} \deg(v) \sum_{i=1}^j (1+\epsilon_1)^i \\
    &\leq \log_{1 + \epsilon_1}(n/l) + \tfrac{2((1+c)\ln(n) + \ln(\log_{1+\epsilon_1}n))(1 + \epsilon_1)^2}{\epsilon_1^2 n} \deg(v) \tfrac{(1+\epsilon_1)^{j+1}}{\epsilon_1} \\
    &= \log_{1 + \epsilon_1}(n/l) + \tfrac{2((1+c)\ln(n) + \ln(\log_{1+\epsilon_1}n))(1 + \epsilon_1)^2}{\epsilon_1^2 n} \deg(v) \tfrac{(1+\epsilon)(n/l)}{\epsilon_1} \\
    &= \log_{1 + \epsilon_1}(n/l) + \tfrac{2((1+c)\ln(n) + \ln(\log_{1+\epsilon_1}n))(1 + \epsilon_1)^3}{\epsilon_1^3l} \deg(v)
\end{align*}
When sampling neighbors, every edge $(u, v)$ may be examined when sampled in line \ref{line:sample} by $u$ or $v$ or both. We can attribute the cost of sampling the endpoints of an edge to one of its endpoints and take twice the cost to get an upper bound on the number of samples made.

First, we consider the nodes $v \in H$ which remain in $H$ by the end of iteration $j$. In iteration $j-1$, $v$ was moved to $L$. For each edge $(u, v)$, either $u$ was moved to $L$ before $v$, in which case we attribute the edge in iteration $j$ to $u$, otherwise we attribute the edge to $v$. The number of remaining edges incident to $v$, which are incident to nodes in $H$ when $v$ was moved to $L$, is $d_H(v)$ at the time of the move (in iteration $j-1$). By Lemma \ref{lem:preconditionsproof}, $d_H(v) < l \cdot (1 + \epsilon_1)^2$ with high probability. Thus, the total number of edges attributed to $v$ is less than $l \cdot (1 + \epsilon)^2$. 

Next, consider the nodes $u \in H$ which were moved to $L$ in iteration $j$. By Lemma \ref{lem:preconditionsproof}, $d_H(u)< l \cdot (1 + \epsilon_1)$. So, the total number of edges attributed to $u$ is less than $l \cdot (1 + \epsilon_1)$. 

Hence, the total number of edges attributed to all the nodes in $H$ is less than $l \cdot (1 + \epsilon_1)^2$. We replace $\deg(v)$ in the equation for the cost of sampling neighbors by the number of attributed edges. Thus, the number of samples of attributed edges for any node in $H$ is less than $\log_{1 + \epsilon_1}(n/l) + \frac{2((1+c)\ln(n) + \ln(\log_{1+\epsilon_1}n))(1 + \epsilon_1)^3}{\epsilon_1^3l} l \cdot (1 + \epsilon_1)^2$ $= \log_{1 + \epsilon_1}(n/l) + \frac{2((1+c)\ln(n) + \ln(\log_{1+\epsilon_1}n))(1 + \epsilon_1)^5}{\epsilon_1^3}$ which we rewrite as $\log_{1 + \epsilon_1}(n/l) + \alpha (\ln(n) + \ln(\log_{1+\epsilon_1}n))$ for a constant $\alpha$ which depends on $\epsilon_1$. We multiply this bound by $n$ to get the bound when summed over all nodes and the total cost of sampling is bound by twice this, which is $2n\log_{1 + \epsilon_1}(n/l) + 2\alpha n(\ln(n) + \ln(\log_{1+\epsilon_1}n))= O(n \log n)$.

Lines 6 to 17 of Algorithm \ref{alg:approxalgorithm} samples neighbors from nodes in $H$ and performs a constant amount of work per sampled neighbor. In particular we note that in the incident-list model, each execution of line \ref{line:sample} requires only constant time. Lines 18 to 25, for each node $u$ in $Q$, perform a constant amount of work for every node that sampled $u$ and each node in $H$ can be added to the $Q$ at most once. Thus, the running time of the randomized part of Algorithm \ref{alg:approxalgorithm} (Lines 3 to 31) is proportional to the number of samples.

To this we need to add the cost of the peeling algorithm at the end if no $l$ is returned.

\begin{lemma}
\label{lem:outcore}
Let $core(k)$ be the set of nodes in a $k$-core for some $k$. Let $outcore(k)=\{(u, v) \in E | u \in core(k), v \in \bigcup_{i\geq k}core(i)\}$. Then $|outcore(k)| \leq k \cdot |core(k)|$.
\end{lemma}

\begin{proof}
In the peeling algorithm, after all nodes with core value $< k$ are peeled, every node remaining has degree at least $k$. When the nodes with core value $k$ are peeled, every node peeled has degree no more than $k$ when it is peeled. Therefore $|outcore(k)| \leq k \cdot |core(k)|$.
\end{proof}

The following lemma follows from Lemma \ref{lem:outcore}.

\begin{lemma}
\label{lem:corenolabel}
The number of edges incident to nodes with core value no greater than $k$ is no greater than $k \cdot |\bigcup_{k'\leq k}core(k')|$.
\end{lemma}

When $p \geq 1$, $l = O(\log n)$. By the proof of Lemma \ref{lem:delta}, the core value of every node is $O(\log n)$ with high probability. By Lemma \ref{lem:corenolabel}, the number of edges incident to every node is $O(\log n)$. Since the peeling algorithm has running time linear in the number of edges, the running time of Line 32 of Algorithm \ref{alg:approxalgorithm} is $O(n \log n)$ with high probability. Therefore, we have the following lemma.

\begin{lemma}
The total running time of Algorithm \ref{alg:approxalgorithm} is $O(n \log n)$.
\end{lemma}

\subsubsection*{Comparing Running Time to the Algorithm from \cite{colton2}}

\label{subsec:coltoncomparison}
The algorithm from \cite{colton2} requires around $\Theta(\frac{n \log n}{\epsilon^2})$ edges in the sampled subgraph, which is obtained in the streaming model by sampling each edge with probability $\Omega(\frac{n \log n}{\epsilon^2 m})$. Directly adapting this method to the incidence list model would require scanning over each edge. However, it can be modified to be more efficient for the incidence list model by sampling a certain number of edges incident to each node according to some probability distribution such that the resulting graph has the required number of edges. To compare the running time of our algorithm against the modified algorithm from \cite{colton2}, we analyze the running time of Algorithm \ref{alg:approxalgorithm} with respect to the nodes that are sampled from. 

\begin{lemma}
\label{lem:totalsamples}
Given a graph $G=(V,E)$ with degeneracy $\delta$, the running time of Algorithm \ref{alg:approxalgorithm}  is $O(n + \frac{|D_s|}{\delta} \log n)$ where $D_s = \sum_{v \in N_s} deg(v)$ and $N_s$ denotes the set of nodes in $V$ with degree $ > \delta / (1 + \epsilon_1)^2$.
\end{lemma}

\begin{proof}
By Lemma \ref{lem:delta}, the output $l$ returned by Algorithm \ref{alg:approxalgorithm} is $> \delta / (1 + \epsilon_1)^2$ with high probability, which means only nodes in $N_s$ with degree $> \delta / (1 + \epsilon_1)^2$ will be added to $H$ and have their neighbours sampled. The number of iterations of the while-loop at Line 3 is $< \log_{1 + \epsilon_1}(\frac{n(1 + \epsilon_1)^2}{\delta})$. Let $p$ denote the initial probability $\frac{2((1 + c) \ln(n) + \ln(\log_{1 + \epsilon_1}(n)))(1 + \epsilon_1)^2}{\epsilon_1^2 n}$. The total number of neighbours sampled is 
\begin{align*}
   &< \sum_{i=0}^{\log_{1 + \epsilon_1}(\frac{n(1 + \epsilon_1)^2}{\delta})} \sum_{v \in N_s} p (1 + \epsilon_1)^i \cdot deg(v) \\
   &= p \cdot D_s \sum_{i=0}^{\log_{1 + \epsilon_1}(\frac{n(1 + \epsilon_1)^2}{\delta})} (1 + \epsilon_1)^i \\
   &= p \cdot D_s \left(\frac{(1 + \epsilon_1)^{(\log_{1 + \epsilon_1}(\frac{n(1 + \epsilon_1)^2}{\delta})+1)}-1}{(1 + \epsilon_1)-1}\right)\\
   &= p \cdot D_s \left(\frac{\frac{n(1 + \epsilon_1)^3}{\delta} - 1}{(1 + \epsilon_1) - 1}\right) <  p \cdot D_s \frac{n(1 + \epsilon_1)^3}{\epsilon_1 \delta}\\
   &= D_s \frac{2((1 + c) \ln(n) + \ln(\log_{1 + \epsilon_1}(n)))(1 + \epsilon_1)^2}{\epsilon_1^2 n} \frac{n(1 + \epsilon_1)^3}{\epsilon_1 \delta} \\
   &= O(\tfrac{D_s}{\delta} \log n)
\end{align*}
Since Algorithm \ref{alg:approxalgorithm} runs proportionally to the number of samples, the running time is $O(\tfrac{D_s}{\delta} \log n)$.
\end{proof}

We use Lemma \ref{lem:totalsamples} to show that Algorithm \ref{alg:approxalgorithm} can have an asymptotically faster running time than the algorithm from \cite{colton2} adapted for the incidence list model, which has a running time of $O(n \log n)$ since it also runs proportionally to the number of sampled edges. Consider an input graph $G$, with $n$ nodes, which consists of the union of disjoint cliques as follows: $G$ has one larger clique of $n^b$ nodes, where $b < 1$, and several other smaller cliques of $n^b/(1 + \epsilon_1)^2$ nodes. We let $B$ denote the set of nodes which belong to the larger clique. The degeneracy of $G$ is $n^b - 1$ and Algorithm \ref{alg:approxalgorithm} only samples neighbours adjacent to the nodes in $B$ which each have degree $n^b -1$. So, we have that $D_s = n^b(n^b - 1)$. By Lemma \ref{lem:totalsamples}, the running time of Algorithm \ref{alg:approxalgorithm} is $O(n^b \log n)$, which is asymptotically less than the running time of the algorithm from \cite{colton2}.

\section{Optimizations}
\label{sec:optimizations}
We present optimizations to improve the running time of Algorithm $\ref{alg:approxalgorithm}$ in practice.

\subsubsection*{Lower Starting Threshold}
The initial degeneracy threshold of $l = \frac{n}{(1 + \epsilon_1)}$ is too high in practice since it is unlikely that the degeneracy of the graph is $n$, and this will cause numerous iterations at the beginning of the algorithm where we end up removing all the nodes from $H$. We can improve this by lowering the starting threshold. Let $d$ be the maximal integer where there are at least $d$ nodes in $G$ with degree $d$. Then, we know that the degeneracy of $G$ is upper bounded by $d$. We can compute $d$ once at the beginning of the algorithm, then divide $l$ by $(1 + \epsilon_1)$ and multiply $p$ by $(1 + \epsilon_1)$ until $l \leq d$. 

\subsubsection*{Larger Threshold Leaps}
We can further reduce the number of iterations required to output a degeneracy by changing $l$ by more in each iteration where $H$ becomes empty. That is, we divide $l$ and multiply $p$ by the sequence $(1 + \epsilon_1), (1 + \epsilon_1)^2, (1 + \epsilon_1)^4, (1 + \epsilon_1)^8, ...$ until we find a $(1 + \epsilon_1)^{2^i}$ for which $H$ is not empty at Line 26. Then, we binary search in the range $(1 + \epsilon_1)^{2^{i-1}}$ to $(1 + \epsilon_1)^{2^i}$ to get the first $k$ where the factor $(1 + \epsilon_1)^k$ results in $H$ not being empty.

\subsubsection*{Reusing Randomness} On Line 10, computing a random neighbor with replacement each time becomes very time consuming in practice. For the random sampling method to work, it does not require a new random number to be computed each time. In the correctness proof, in the proof of Lemma \ref{lem:probbounds}, since we take the union bound on the probability of error over all nodes and iterations, we can reuse the same randomness for sampling from all nodes in all iterations. At the beginning of the algorithm, we initialize an array $R$ of size $md(V)$, where $md(V)$ is the maximum degree of any vertex in $V$ of random numbers within the range $[0, 1)$. Then, each time we need to sample the $i^{th}$ neighbour of a node $v$ (Line 10), we take the $(R_i \cdot deg(v))^{th}$ element from $v$'s incidence list. 
Since $k(v)$ increases in each iteration, we remember the sampled neighbors for $v$ from the previous iteration and add to it as needed. 

\section{Experiments}

\begin{table}
\begin{tabular}{|c|r|r|}
\hline
Graph        & \multicolumn{1}{c|}{Nodes} & \multicolumn{1}{c|}{Edges} \\ \hline
clueweb12    & 978,408,098                & 74,744,358,622             \\ \hline
uk-2014      & 787,801,471                & 85,258,779,845             \\ \hline
gsh-2015     & 988,490,691                & 51,984,719,160             \\ \hline
sk-2005      & 50,636,154                 & 3,639,246,313              \\ \hline
twitter-2010 & 41,652,230                 & 2,405,026,390              \\ \hline
\end{tabular}
\caption{Summary of datasets}
\label{tab:datasets}
\end{table}

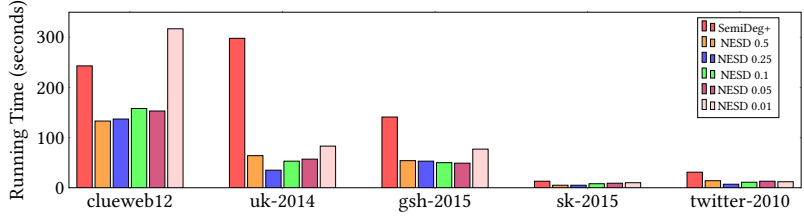
\begin{figure*}
    \centering
    \subfigure{
    \resizebox{0.6\textwidth}{!}{
    \begin{tikzpicture}
      \centering
      \begin{axis}
      [
        ybar,
        ymin=0,
        ymax=350,
        ticks=both,
        width = 20cm, 
        height = 6cm,
        bar width=11pt,
        label style={font=\huge},
        tick label style={font=\huge},
        ylabel={Running Time (seconds)},
        legend pos=north east,
        x axis line style = { opacity = 1 },
        xtick distance = 1, 
        tickwidth  = 1pt,
        symbolic x coords = {clueweb12, uk-2014, gsh-2015, sk-2015, twitter-2010},
      ]
    
    \addplot[fill = red!65] coordinates {(clueweb12, 243) (uk-2014,298) (gsh-2015, 141) (sk-2015, 13) (twitter-2010, 31)};
    \addplot[fill = orange!70] coordinates {(clueweb12, 133) (uk-2014,64) (gsh-2015, 54) (sk-2015, 5) (twitter-2010, 14)};
    \addplot[fill = blue!65] coordinates {(clueweb12, 137) (uk-2014,35) (gsh-2015, 53) (sk-2015, 5) (twitter-2010, 7)};
    \addplot[fill = green!60] coordinates {(clueweb12, 158) (uk-2014, 53) (gsh-2015, 50) (sk-2015, 8) (twitter-2010, 11)};
    \addplot[fill = purple!65] coordinates {(clueweb12, 153) (uk-2014, 57) (gsh-2015, 49) (sk-2015, 9) (twitter-2010, 13)};
    \addplot[fill = pink!65] coordinates {(clueweb12, 317) (uk-2014, 83) (gsh-2015, 77) (sk-2015, 10) (twitter-2010, 12)};
    
    \legend{SemiDeg+, NESD 0.5, NESD 0.25, NESD 0.1, NESD 0.05, NESD 0.01}
    \end{axis}
    \end{tikzpicture}
    }

}
    \caption{\label{fig:runtime} Running time comparison between SemiDeg+ and NESD with various $\epsilon$ and $c=0.5$}
    \label{fig:runtimes}
\end{figure*}

\begin{table*}
\begin{minipage}{.102\linewidth}
\begin{tabular}{l}
                                \\ \hline
\multicolumn{1}{|l|}{SemiDeg+}  \\ \hline
\multicolumn{1}{|l|}{NESD 0.5}  \\ \hline
\multicolumn{1}{|l|}{NESD 0.25} \\ \hline
\multicolumn{1}{|l|}{NESD 0.1}  \\ \hline
\multicolumn{1}{|l|}{NESD 0.05} \\ \hline
\multicolumn{1}{|l|}{NESD 0.01} \\ \hline
\end{tabular}
\caption*{ }
\end{minipage}%
\begin{minipage}{.145\linewidth}
\begin{tabular}{|l|l|}
\hline
Output   & Error \\ \hline
4244     &       \\ \hline
3696.835 & 0.871 \\ \hline
4003.074 & 0.943 \\ \hline
4186.726 & 0.987 \\ \hline
4186.726 & 0.984 \\ \hline
4238.685 & 0.999 \\ \hline
\end{tabular}
\caption*{\textmd{clueweb12}}
\end{minipage}%
\begin{minipage}{.145\linewidth}
\begin{tabular}{|l|l|}
\hline
Output   & Error \\ \hline
10003    &       \\ \hline
8756.98  & 0.875 \\ \hline
9884.616 & 0.988 \\ \hline
9947.312 & 0.994 \\ \hline
9850.165 & 0.985 \\ \hline
9998.635 & 0.999 \\ \hline
\end{tabular}
\caption*{\textmd{uk-2014}}
\end{minipage}%
\begin{minipage}{.145\linewidth}
\begin{tabular}{|l|l|}
\hline
Output   & Error \\ \hline
9888     &       \\ \hline
9418.101 & 0.952 \\ \hline
9004.657 & 0.911 \\ \hline
9601.476 & 0.971 \\ \hline
9645.406 & 0.975 \\ \hline
9645.45  & 0.975 \\ \hline
\end{tabular}
\caption*{\textmd{gsh-2015}}
\end{minipage}%
\begin{minipage}{.145\linewidth}
\begin{tabular}{|l|l|}
\hline
Output   & Error \\ \hline
4510     &       \\ \hline
4175.481 & 0.926 \\ \hline
4338.047 & 0.962 \\ \hline
4425.227 & 0.981 \\ \hline
4452.054 & 0.987 \\ \hline
4502.467 & 0.998 \\ \hline
\end{tabular}
\caption*{\textmd{sk-2005}}
\end{minipage}%
\begin{minipage}{.145\linewidth}
\begin{tabular}{|l|l|}
\hline
Output   & Error \\ \hline
2488     &       \\ \hline
2162.936 & 0.869 \\ \hline
2391.45  & 0.961 \\ \hline
2376.771 & 0.955 \\ \hline
2545.711 & 1.023 \\ \hline
2534.374 & 1.019 \\ \hline
\end{tabular}
\caption*{\textmd{twitter-2010}}
\end{minipage}%
\caption{\label{tab:errorfactors} Error factors of NESD degeneracy outputs with various $\epsilon$}
\end{table*}

We refer to our approach as NESD (Neighbor Sampling for Degeneracy). In our experiments, we compare NESD (including all the optimizations from Section \ref{sec:optimizations}) to the approximation algorithm from \cite{colton2} and the state-of-the-art exact degeneracy algorithm, SemiDeg+, from \cite{semideg}. Since the goal of our approach is to achieve an efficient sublinear $(1 + \epsilon)$ approximation algorithm, we do not compare versus other less accurate approximate algorithms nor other exact algorithms since they are superseded by SemiDeg+. We experiment with 5 different values of $\epsilon$, namely 0.5, 0.25, 0.1, 0.05, 0.01 to illustrate the varying approximation accuracies of our algorithm. The value of $c$ did not have a major impact on the results since the probability of correctness is high with $n$ being large in all of the datasets. So, we use $c=0.5$ for all experiments.

We use the datasets twitter-2010, sk-2005, gsh-2015, uk-2014, and clueweb12 from the Laboratory of Web Algorithmics \cite{LWA1, LWA2, LWAbubing}, and the characteristic of each graph is illustrated in Table \ref{tab:datasets}. We note that the number of edges shown in Table \ref{tab:datasets} is after symmetrization, where we add the reverse of directed edges if they do not already exist, and without any self loops. Each of the graphs are massive webgraphs with over 1 billion edges, and the largest webgraph, cluweb12, has over 74 billion edges. The algorithms used in the experiments were implemented in Java 11, and the experiments were run on Amazon EC2 instances with the r5.24xlarge instance size (96 vCPU, 768GB memory) running Linux.

Figure \ref{fig:runtimes} shows the running time results of NESD with each $\epsilon$. On each of the datasets, the approximation algorithm from \cite{colton2} was orders of magnitude slower than our algorithm and SemiDeg+ due to the requirement of passing through every single edge of the graph. As such, when run on massive webgraphs, the running time was extremely slow and we do not include the results in the figure.  NESD ran 2.17x to 4.56x faster than SemiDeg+ on twitter-2010, 1.33x to 2.61x faster on sk-2005, 1.81x to 2.86x faster on gsh-2015, and 3.57x to 8.45x faster on uk-2014. On clueweb12, NESD ran 1.53x to 1.83x faster with $\epsilon$ between $0.5$ to $0.05$, but 1.3x slower on $\epsilon=0.01$, which is the only input in all the experiments where NESD was slower.

Table \ref{tab:errorfactors} shows the degeneracy $\delta$ of each dataset (output by SemiDeg+) as well as the approximate degeneracy and error factors from the output by NESD with each $\epsilon$. The error factors on each dataset were between $0.869$ and $0.952$ for $\epsilon = 0.5$, between $0.911$ and $0.988$ for $\epsilon=0.25$, between $0.955$ and $0.994$ for $\epsilon=0.1$, between $0.975$ and $1.023$ for $\epsilon = 0.05$, and between $0.975$ and $1.019$ for $\epsilon=0.01$. This shows that the error factors are well within the proven theoretical bounds and it also demonstrates the high level of degeneracy approximation accuracy of our algorithm.

\section{Extension to \texorpdfstring{$k$}{k}-Core Decomposition}
We can extend Algorithm \ref{alg:approxalgorithm} to solve full $k$-core decomposition. In this section, we describe how we can can modify the algorithm to label each node in the graph with an approximate core number within a factor of $(1 + \epsilon)$ of its true core number with high probability in $O(n \log n)$ time.

\subsection{Algorithm Modifications}
Instead of returning $l$ as the approximate degeneracy when $H$ is non-empty on Line 31, we continue the algorithm and assign an approximate core number label $l$ nodes which remain in $H$ (Line 30). Within the body of the else statement on Line 30, we label all the nodes $v \in H$ with the value of $l$ of that iteration. We then add the line of code $V \leftarrow L$ such that in the next iteration, the vertices in $V$ are only the unlabelled nodes. We then also perform the same updates to $l$, $p$ and $Sampled(v)$ for all $v\in V$. 

On Line 32, when $p \geq 1$, we run the peeling algorithm on the remaining unlabelled nodes $V$ and assign core numbers to nodes which were not labeled inside the while-loop. Let $l'$ be the last approximate label assigned within the while-loop. If the peeling algorithm assigns a core number of $>l'/(1 + (3/2) \epsilon_1)$ to a node, we instead label that node with $l'$. When all nodes have an assigned label, we return those labels as $(1+\epsilon)$-approximate core numbers for each node.

\subsection{Correctness}
The high probability of correctness for the approximate core values follow from the Lemmas \ref{lem:uniformsample} to \ref{lem:correctness} in Section \ref{subsec:degeneracycorrectness}. However, we need to prove that using the peeling algorithm as described for labelling the unlabeled nodes satisfies the desired approximation bounds. This is a concern because nodes with lower core values may have already been labelled and removed from the graph of unlabelled nodes, and therefore will not participate in the peeling algorithm. Consequently it is possible that nodes in the peeling algorithm will be assigned a higher core value than their actual core value.

\begin{lemma}
Running the peeling algorithm on the unlabelled nodes (as described above) assigns $(1 + \epsilon)$-approximate core numbers to those nodes.
\end{lemma}

\begin{proof}
Let $V'$ be the set of unlabelled nodes and let $l'$ be the last label assigned within the while-loop (described above). Consider the nodes in $V'$ which are assigned a label of $< l'/(1 + (3/2)\epsilon_1)$ by running the peeling algorithm on $V'$. By Lemma \ref{lem:delta}, all nodes with core value of $< l'/(1 + (3/2) \epsilon_1)$ would also be unlabelled. Hence, the nodes which are assigned a label of $< l'/(1 + (3/2) \epsilon_1)$ by the peeling algorithm receive the same label as they would in the peeling algorithm run on the entire graph. 

Next, consider the nodes $v \in V'$ which would have received a label of $> l'/(1 + (3/2) \epsilon_1)$ by running the peeling algorithm on $V'$, but we label them by $l'$. By the proof of Lemma \ref{lem:delta}, since $v$ was moved to $L$ in the last iteration, $c(v)< l' \cdot (1 + \epsilon_1)$. Hence, the approximation bounds hold if we label them with $l'$.
\end{proof}
\subsection{Running Time}
We slightly modify the analysis of the running time in Section \ref{subsec:degeneracyrunningtime} to consider the time it takes to label all nodes with an approximate core number. Let $l(v)$ denote the label assigned to any node $v$. The neighbors of $v$ are sampled in each iteration until does not get moved to $L$ and gets assigned a label. The maximum number of iterations of the while-loop is $j = \log_{1 + \epsilon_1} n$ which means the total cost of sampling node $v$ is $\log_{1 + \epsilon_1}(n) + \tfrac{2((1+c)\ln(n) + \ln(\log_{1+\epsilon_1}n))(1 + \epsilon_1)^3}{\epsilon_1^3} \deg(v)$. 


Similar to Section \ref{subsec:degeneracyrunningtime}, consider the iteration where node $v$ is labelled $l(v)$ ($v$ was moved to $L$ in the previous iteration). In the previous iteration, for each edge $(u, v)$, if $u$ was moved to $L$ before $v$, we attribute the cost of sampling that edge in iteration $j$ to $u$. Otherwise, we attribute that edge to $v$. This time, we consider the edges attributed to the nodes which are assigned a label in each iteration, and by the same analysis as in Section \ref{subsec:degeneracyrunningtime}, the edges attributed for nodes with labels $l(v)$ in any iteration is $< l(v) ( 1 + \epsilon)^2$. So, the number of samples of attributed edges for any node with label $l(v)$ in any iteration is less than $\log_{1 + \epsilon_1}(n) + \alpha (\ln(n) + \ln(\log_{1+\epsilon_1}n))$ for some constant $\alpha$. Since this quantity is the same for all labels, and each node is assigned one label, this bound when summed over all nodes is  $n\log_{1 + \epsilon_1}(n) + \alpha n (\ln(n) + \ln(\log_{1+\epsilon_1}n))$. Then, the total cost of sampling is bound by twice this which is $2n\log_{1 + \epsilon_1}(n) + 2\alpha n (\ln(n) + \ln(\log_{1+\epsilon_1}n)) = O(n \log n)$.

Finally, the analysis of the running time of the peeling algorithm is the same as the analysis from Section \ref{subsec:degeneracyrunningtime}. Therefore, Theorem \ref{t:main} follows.

\section{Conclusion}
We have devised $O(n\log n)$ algorithms for degeneracy and for $k$-core in the incident-list model. Both algorithms give provably close approximate solutions, with high probability. We have shown in experiments, that on all our instances of very large graphs, our algorithm for degeneracy, NESD works within our approximation guarantees, and for almost all our data, it gives a significant speed up over prior state-of-the-art methods for computing degeneracy. Our algorithm features various optimizations which may be of independent interest. 

 \newpage
 \balance
 \bibliography{bibliography}
 \bibliographystyle{acm}
\end{document}